\documentclass[12pt]{article}
\usepackage{amsmath, amssymb, amsfonts}
\usepackage{amsthm}
\usepackage[utf8]{inputenc}
\usepackage{csquotes}
\usepackage{geometry}
\usepackage{adjustbox}
\usepackage{graphicx}
\usepackage{verbatim}
\usepackage{array}
\usepackage{enumitem}
\usepackage{booktabs}
\usepackage{hyperref}
\usepackage{url}
\usepackage[capitalize]{cleveref}
\usepackage{tikz}
\usetikzlibrary{positioning, shapes.geometric, arrows.meta}

\title{Invariant-Based Cryptography: Toward a General Framework}
\author{Stanislav Semenov \\
\href{mailto:stas.semenov@gmail.com}{stas.semenov@gmail.com} \\
\href{https://orcid.org/0000-0002-5891-8119}{ORCID: 0000-0002-5891-8119}}
\date{May 12, 2025}

\theoremstyle{definition}
\newtheorem{definition}{Definition}[section]

\theoremstyle{plain}

\newtheorem{theorem}[definition]{Theorem}

\theoremstyle{remark}

\begin{document}

\maketitle

\begin{abstract}
We develop a generalized framework for invariant-based cryptography by extending the use of structural identities as core cryptographic mechanisms. Starting from a previously introduced scheme where a secret is encoded via a four-point algebraic invariant over masked functional values, we broaden the approach to include multiple classes of invariant constructions. In particular, we present new symmetric schemes based on shifted polynomial roots and functional equations constrained by symmetric algebraic conditions, such as discriminants and multilinear identities. These examples illustrate how algebraic invariants—rather than one-way functions—can enforce structural consistency and unforgeability. We analyze the cryptographic utility of such invariants in terms of recoverability, integrity binding, and resistance to forgery, and show that these constructions achieve security levels comparable to the original oscillatory model. This work establishes a foundation for invariant-based design as a versatile and compact alternative in symmetric cryptographic protocols.
\end{abstract}

\subsection*{Mathematics Subject Classification}
03F60 (Constructive and recursive analysis), 94A60 (Cryptography)

\subsection*{ACM Classification}
F.4.1 Mathematical Logic, E.3 Data Encryption

\section*{Introduction}

Invariant-based cryptography (IBC) is a recently introduced paradigm in symmetric cryptographic design, in which structural identities---\emph{invariants}---replace traditional one-way functions as the core source of security. In contrast to algebraic hardness assumptions that rely on inversion difficulty, IBC schemes derive robustness from the \emph{preservation} of a deterministic relation among data points. Any deviation from the intended structure breaks the invariant, enabling implicit integrity checks and recovery logic without revealing secrets.

A concrete realization of this idea was presented in~\cite{semenov2025invariant}, where a symmetric cryptographic primitive was constructed around an exact four-point functional identity over discretized oscillatory functions. That construction demonstrated that a secret index, hidden within an evaluation grid, could be cryptographically protected not through concealment but through the impossibility of structurally consistent extension. The resulting scheme was compact, invertible only with internal alignment parameters, and resistant to adversarial manipulation.

The present work generalizes that approach. Instead of relying on a single analytic invariant, we introduce a flexible framework that incorporates a broader class of algebraic structures---such as discriminants of polynomials and projective cross-ratios---as cryptographic invariants. These constructions preserve the core idea of invariant-based security while broadening its scope and applicability. We describe a range of symmetric schemes, analyze their recoverability and integrity properties, and show how invariant reuse enables efficient session modes.

Our contribution is threefold:
\begin{enumerate}
    \item We extract the underlying design principles from the oscillatory model and express them in terms of algebraic invariants acting over masked or shared components;
    \item We construct new symmetric protocols based on discriminants and cross-ratios, each enforcing consistency via a structurally rigid identity;
    \item We analyze the cryptographic implications of invariant reuse, indistinguishability under masking, and the design of lightweight verification logic.
\end{enumerate}

Taken together, these developments define a general framework for invariant-based cryptography---compact, self-validating, and structurally expressive. The result is a novel design axis in symmetric cryptographic engineering: \emph{not what is hidden, but what is preserved}.

\section{Symmetric Discriminant-Based Scheme}

One of the most classical examples of an algebraic invariant is the \emph{discriminant} of a polynomial. Originating in 19th-century algebraic analysis~\cite{gelfand1994discriminants}, the discriminant provides a compact expression that vanishes if and only if a polynomial has repeated roots. More generally, it encodes the symmetric relations among a set of roots and remains invariant under permutation and certain coordinate transformations.

In this section, we apply the discriminant as a cryptographic invariant within a symmetric protocol. By constructing a cubic polynomial whose roots include a hidden component, we ensure that the transmitted data satisfy a structural constraint that is both verifiable and nontrivial to forge. The security of the scheme arises not from concealing the polynomial itself, but from the difficulty of reconstructing its root structure from partial information.

This use of the discriminant continues the central theme of invariant-based cryptography: leveraging preserved algebraic identities to enforce structural consistency, integrity, and recoverability. The result is a lightweight symmetric scheme in which a secret offset is implicitly encoded within a polynomial, and retrieved through constrained evaluation and algebraic recovery.

\subsection{Common Setup}

\begin{itemize}
  \item \textbf{Public parameters:}
    \begin{itemize}
      \item A fixed field modulus \( M \in \mathbb{N} \), typically a 256-bit prime;
      \item A shared evaluation function \( t = t(S, z) \in \mathbb{Q} \), derived from session nonce and shared secret;
      \item A secure hash function \( H \colon \{0,1\}^* \to \mathbb{Z}_M \), used for integrity binding;
      \item Cubic polynomials \( P(x) \in \mathbb{Z}_M[x] \) defined by their roots and used to encode algebraic invariants.
    \end{itemize}
  \item \textbf{Shared secret:} A 256-bit string \( S \in \{0,1\}^{256} \), known to both parties.
\end{itemize}

\subsection{Alice’s Generation}

Given a session-specific nonce \( z \in \{0,1\}^{256} \), Alice performs:

\begin{enumerate}
  \item Derives the rational evaluation point \( t := t(S, z) \in \mathbb{Q} \), computed via a PRF or hash-based derivation.

  \item Selects:
    \begin{itemize}
      \item Random secret offset \( h \in \mathbb{Z}_M \);
      \item Three distinct roots \( a_1, a_2, a_3 \in \mathbb{Z}_M \), where \( a_1 \) is kept private and \( a_2, a_3 \) will be transmitted.
    \end{itemize}

  \item Constructs the cubic polynomial:
    \[
    P(x) = (x - a_1)(x - a_2)(x - a_3) \in \mathbb{Z}_M[x].
    \]

  \item Computes:
    \begin{itemize}
      \item Discriminant \( D := \Delta(P) \), an invariant of the root structure;
      \item Shifted evaluation \( y := P(t + h) \in \mathbb{Z}_M \), hiding the shift \( h \);
      \item Binding hash \( H_{\mathrm{check}} := H(S, z, a_2, a_3, D, y) \);
      \item Optional second-level hash \( H_{\mathrm{auth}} := H(S, z, h) \).
    \end{itemize}

  \item Sends to Bob:
    \[
        \langle a_2,\ a_3,\ D,\ y,\ z,\ H_{\mathrm{check}} \rangle
        \quad \text{with optional } H_{\mathrm{auth}}.
    \]
\end{enumerate}

\subsection{\texorpdfstring{Bob’s Verification and Recovery of \( h \)}{Bob's Verification and Recovery of h}}

\begin{enumerate}
  \item Computes \( t := t(S, z) \in \mathbb{Q} \) deterministically from the shared secret and nonce.
  
  \item Verifies integrity of transmitted data:
    \[
    H(S, z, a_2, a_3, D, y) \stackrel{?}{=} H_{\mathrm{check}}.
    \]
    If this fails, abort immediately.

  \item Given \( a_2, a_3, D \), reconstructs candidate values for \( a_1 \) by solving the discriminant identity:
    \[
    \Delta(a_1, a_2, a_3) = D.
    \]
    This typically yields up to two valid candidates.

  \item For each candidate \( a_1 \), reconstructs the polynomial:
    \[
    P(x) = (x - a_1)(x - a_2)(x - a_3).
    \]

  \item Forms the equation:
    \[
    P(t + h) = y,
    \]
    where all quantities are known except \( h \). Expanding the left-hand side yields a cubic equation in \( h \).

  \item Solves this equation in \( \mathbb{Z}_M \) to recover valid roots \( h_i \).\\
    If \( H_{\mathrm{auth}} \) is present, disambiguation is performed via:
    \[
    H(S, z, h_i) \stackrel{?}{=} H_{\mathrm{auth}}.
    \]
    The correct \( h \) is accepted only if exactly one root matches.\\
    If \( H_{\mathrm{auth}} \) is absent, all valid roots are considered acceptable.
\end{enumerate}

\subsection{Remarks on Generalization}

The scheme presented above uses the discriminant of a cubic polynomial as a concrete and accessible example of an algebraic invariant applied in a symmetric cryptographic setting. However, this choice is far from exhaustive. In principle, analogous constructions can be built using polynomials of higher degree, leading to richer invariant structures and more complex root configurations.

Moreover, future protocols may consider distributed variants, in which different participants know different components of a shared polynomial---for instance, different root subsets or partial coefficient sets. This opens the door to schemes where structural invariants link knowledge held across multiple parties, enabling threshold verification, cross-authentication, or collaborative key recovery.

While such generalizations are conceptually appealing, we deliberately refrain from abstracting them in full generality here. Instead, we leave this direction as an open path for future exploration, emphasizing that the discriminant-based method shown above is only one instance of a broader class of invariant-driven designs.

\section{Session Modes with Invariant Reuse}

Invariant-based schemes allow for flexibility in how root triples and transmitted values are generated and interpreted across sessions. Two primary modes of operation can be distinguished depending on whether the invariant (e.g., the discriminant of the cubic polynomial) is re-transmitted for each message or reused throughout a session.

\subsection{Derived Invariant Mode (Minimal Transmission)}

In this mode, both parties derive a shared invariant value \( D \in \mathbb{Z}_M \) deterministically from the session parameters:
\[
D := H_{\mathrm{inv}}(S, z),
\]
where \( H_{\mathrm{inv}} \) is a cryptographic hash function with domain separation. This invariant defines the discriminant constraint on root triples used in the session.

For each message, Alice performs the following:

\begin{enumerate}
  \item Samples a random root triple \( (a_1, a_2, a_3) \in \mathbb{Z}_M^3 \) such that
  \[
  \Delta(a_1, a_2, a_3) = D.
  \]
  
  \item Chooses a random offset \( h \in \mathbb{Z}_M \), constructs the cubic polynomial
  \[
  P(x) = (x - a_1)(x - a_2)(x - a_3),
  \]
  and evaluates
  \[
  y := P(t + h),
  \]
  where \( t = t(S, z) \in \mathbb{Q} \) is the shared evaluation point.

  \item Sends only \( \langle a_2,\ a_3,\ y \rangle \) to Bob.
\end{enumerate}

Bob, upon receiving the message, performs:

\begin{enumerate}
  \item Recomputes \( t := t(S, z) \) and \( D := H_{\mathrm{inv}}(S, z) \).

  \item Solves \( \Delta(a_1, a_2, a_3) = D \) to recover the matching root \( a_1 \) consistent with the fixed invariant.

  \item Constructs \( P(x) \) and solves \( P(t + h) = y \) to recover \( h \).
\end{enumerate}

This mode enables compact messages and invariant enforcement without repeated transmission of \( D \), while maintaining full recoverability of the embedded value.

\subsection{Shared Root Mode (Invariant-Free Continuation)}

In this variant, the invariant is used only in the initial transmission to allow Bob to recover a hidden root \( a_1 \), after which both parties operate with a shared, fixed \( a_1 \) in subsequent messages. This eliminates the need to transmit or enforce the invariant across the session and simplifies future communication.

\paragraph{Initialization phase.}
Alice generates a full root triple \( (a_1, a_2, a_3) \) with discriminant \( D := \Delta(a_1, a_2, a_3) \), chooses a random offset \( h \), and sends:
\[
\langle a_2,\ a_3,\ D,\ y \rangle,
\quad \text{where } y := P(t + h), \quad P(x) = (x - a_1)(x - a_2)(x - a_3).
\]
Bob receives the values, reconstructs \( a_1 \) from the discriminant condition, and solves \( P(t + h) = y \) to find \( h \).  
Now, both parties share the value of \( a_1 \) and agree to keep it fixed throughout the session.

\paragraph{Streaming phase.}
For each subsequent message, Alice:
\begin{enumerate}
  \item Picks new open roots \( a_2, a_3 \in \mathbb{Z}_M \) and a new random offset \( h \in \mathbb{Z}_M \);
  \item Builds the polynomial \( P(x) = (x - a_1)(x - a_2)(x - a_3) \);
  \item Sends only:
  \[
  \langle a_2,\ a_3,\ h \rangle.
  \]
\end{enumerate}

Bob, knowing \( t \), \( a_1 \), and receiving \( a_2, a_3, h \), performs:
\begin{enumerate}
  \item Reconstructs \( P(x) \) using shared \( a_1 \) and received \( a_2, a_3 \);
  \item Computes the shared value:
  \[
  y := P(t + h),
  \]
  which can serve as a symmetric session secret, ephemeral key, or cryptographic tag.
\end{enumerate}

\paragraph{Note.}
This mode eliminates the need for discriminant constraints or invariant recovery after initialization. It is compact and efficient but requires trusting that \( a_1 \) was correctly recovered and agreed upon during setup.

\section{Symmetric Cross-Ratio Scheme}

The \emph{cross-ratio} is a fundamental invariant in projective geometry, preserved under all Möbius (or fractional linear) transformations. Defined for ordered quadruples of points on the projective line, it captures a canonical relation that remains unchanged under the full action of the group $\mathrm{PGL}_2(\mathbb{Z}_M)$. As such, the cross-ratio has long played a central role in classical geometry, complex analysis, and invariant theory~\cite{coxeter2003projective}.

In this section, we adapt the cross-ratio as a cryptographic invariant within a symmetric scheme. The central idea is to hide one element of a quadruple whose cross-ratio is fixed and recoverable only by parties with knowledge of a shared secret. Projective masking is employed to obscure the visible components, ensuring that transmitted data reveals no structural alignment while preserving invariant verifiability.

This construction continues the invariant-based paradigm by demonstrating that projective consistency---rather than arithmetic secrecy---can serve as the basis for symmetric coordination, integrity enforcement, and data binding.

\subsection{Common Setup}

\begin{itemize}
  \item \textbf{Public parameters:}
    \begin{itemize}
      \item A fixed prime modulus \( M \in \mathbb{N} \), typically 256 bits;
      \item A secure hash function \( H \colon \{0,1\}^* \to \mathbb{Z}_M^\times \), used to derive the session invariant;
      \item An optional masking map \( f(z) = \frac{az + b}{cz + d} \in \mathrm{PGL}_2(\mathbb{Z}_M) \), invertible and known only to legitimate parties;
      \item A canonical expression for the cross-ratio:
        \[
        \mathrm{CR}(z_1, z_2;\ z_3, z_4) := \frac{(z_1 - z_3)(z_2 - z_4)}{(z_1 - z_4)(z_2 - z_3)} \mod M.
        \]
    \end{itemize}
  \item \textbf{Shared secret:} A 256-bit string \( S \in \{0,1\}^{256} \), known to both parties.
\end{itemize}

\subsection{Alice’s Generation}

Given a session-specific nonce \( z \in \{0,1\}^{256} \), Alice proceeds as follows:

\begin{enumerate}
  \item Derives the session invariant:
    \[
    I := H(S, z) \in \mathbb{Z}_M^\times.
    \]

  \item Selects arbitrary but distinct elements \( z_1, z_2, z_3 \in \mathbb{Z}_M \), ensuring denominators in CR do not vanish.

  \item Solves the equation:
    \[
    \mathrm{CR}(z_1, z_2;\ z_3, z_4) = I,
    \]
    to compute the unique \( z_4 \in \mathbb{Z}_M \) satisfying the relation.

  \item Applies optional projective masking (if used):
    \[
    \tilde{z}_i := f(z_i),\quad \text{for } i = 1, 2, 3.
    \]

  \item Sends to Bob:
    \[
    \langle \tilde{z}_1,\ \tilde{z}_2,\ \tilde{z}_3,\ z \rangle.
    \]
\end{enumerate}

\subsection{\texorpdfstring{Bob’s Recovery of \( z_4 \)}{Bob's Recovery of z4}}

\begin{enumerate}
  \item Computes the session invariant:
    \[
    I := H(S, z).
    \]

  \item Applies inverse masking (if used):
    \[
    z_i := f^{-1}(\tilde{z}_i),\quad \text{for } i = 1, 2, 3.
    \]

  \item Solves for \( z_4 \in \mathbb{Z}_M \) the cross-ratio equation:
    \[
    \frac{(z_1 - z_3)(z_2 - z_4)}{(z_1 - z_4)(z_2 - z_3)} = I.
    \]
    This is a rational equation linear in \( z_4 \), solvable by:
    \[
    z_4 = \frac{
      (z_1 - z_3)(z_2) - I(z_2 - z_3)(z_1)
    }{
      (z_1 - z_3) - I(z_2 - z_3)
    } \mod M,
    \]
    assuming the denominator is invertible modulo \( M \).

  \item Uses \( z_4 \) as a session-specific shared value or input to a key derivation function.
\end{enumerate}

\paragraph{Projective masking via \(\mathrm{PGL}_2(\mathbb{Z}_M)\).}

To protect the transmitted values \( z_i \) from structural analysis, we apply a shared masking transformation drawn from the projective general linear group \( \mathrm{PGL}_2(\mathbb{Z}_M) \). This group consists of all invertible Möbius (fractional linear) transformations over the field \( \mathbb{Z}_M \), defined by:

\[
f(z) = \frac{a z + b}{c z + d}, \quad \text{where } a,b,c,d \in \mathbb{Z}_M,\ ad - bc \not\equiv 0 \mod M.
\]

Two such transformations are considered equivalent if they differ by a nonzero scalar multiple; that is, \( f \sim \lambda f \) for \( \lambda \in \mathbb{Z}_M^\times \). The set of equivalence classes under this relation forms the group \( \mathrm{PGL}_2(\mathbb{Z}_M) \), which acts on the projective line \( \mathbb{P}^1(\mathbb{Z}_M) = \mathbb{Z}_M \cup \{\infty\} \).

These transformations preserve the cross-ratio:
\[
\mathrm{CR}(f(z_1), f(z_2); f(z_3), f(z_4)) = \mathrm{CR}(z_1, z_2; z_3, z_4),
\]
making them ideal for cryptographic masking: the functional invariant remains unchanged, while the inputs become indistinguishable from pseudorandom noise.

In this scheme, the transformation \( f \in \mathrm{PGL}_2(\mathbb{Z}_M) \) is derived from the shared secret and session nonce, for example:
\[
f := \left(\frac{a z + b}{c z + d}\right), \quad \text{where } (a, b, c, d) := H_{\mathrm{mask}}(S, z) \mod M,\ ad - bc \ne 0.
\]
The inverse transformation \( f^{-1} \) is known to both parties, allowing recovery of the original points before cross-ratio evaluation.

\paragraph{Integrity binding (optional).}
To prevent tampering, Alice may also include a hash tag:
\[
H_{\mathrm{check}} := H(S, z, \tilde{z}_1, \tilde{z}_2, \tilde{z}_3),
\]
which Bob verifies before using the result.

\begin{theorem}[Invariant Indistinguishability under Projective Masking]
Let \( M \) be a large prime modulus, and let \( I \in \mathbb{Z}_M^\times \) be a fixed value. Suppose that for each session \( i = 1, \dots, N \), a cross-ratio identity holds:
\[
\mathrm{CR}(z_1^{(i)}, z_2^{(i)};\ z_3^{(i)}, z_4^{(i)}) = I,
\]
and that only the masked triples
\[
(\tilde{z}_1^{(i)}, \tilde{z}_2^{(i)}, \tilde{z}_3^{(i)}), \quad \tilde{z}_j^{(i)} := f_i(z_j^{(i)}), \quad j = 1,2,3,
\]
are observable to an adversary, where each \( f_i \in \mathrm{PGL}_2(\mathbb{Z}_M) \) is independently and uniformly chosen.

Then, in the absence of knowledge of \( z_4^{(i)} \), any polynomial-time adversary has negligible advantage in recovering or distinguishing the invariant \( I \), even given unbounded access to the masked triples.
\end{theorem}

\begin{proof}[Sketch of proof.]
The cross-ratio function
\[
\mathrm{CR}(z_1, z_2;\ z_3, z_4)
\]
is well-defined only on quadruples and is invariant under the action of \( \mathrm{PGL}_2(\mathbb{Z}_M) \). In particular, for any invertible projective map \( f \), we have:
\[
\mathrm{CR}(f(z_1), f(z_2); f(z_3), f(z_4)) = \mathrm{CR}(z_1, z_2; z_3, z_4).
\]

Thus, the masked triple \( (\tilde{z}_1, \tilde{z}_2, \tilde{z}_3) \) represents the original points up to arbitrary coordinate change (projective basis). Since \( f_i \) is independently chosen for each session, the observed masked triples lie in unrelated projective frames. Without a common reference point (such as \( z_4^{(i)} \)), the cross-ratio becomes uncomputable.

Moreover, for each session, the mapping from the secret quadruple \( (z_1^{(i)}, z_2^{(i)}, z_3^{(i)}, z_4^{(i)}) \) to the masked triple \( (\tilde{z}_1^{(i)}, \tilde{z}_2^{(i)}, \tilde{z}_3^{(i)}) \) is information-losing with respect to \( I \). No algebraic relation exists between three projectively transformed points and the invariant, as projective transformations can arbitrarily reposition triples while preserving cross-ratio only when the fourth point is available.

Therefore, the distribution of observed masked triples is independent of \( I \), rendering it indistinguishable across sessions.
\end{proof}

\section{Extension to Finite Fields and Algebraic Structures}

All invariant-based schemes presented in this work are defined algebraically and make no structural assumption about the base ring beyond arithmetic closure. Consequently, they can be instantiated not only over prime fields \( \mathbb{Z}_p \), but also over:

\begin{itemize}
  \item Modular rings \( \mathbb{Z}_M \), where \( M \) is composite;
  \item Finite fields \( \mathbb{F}_{p^n} \), constructed via irreducible polynomials;
  \item Finite algebras, such as vector spaces over fields with additional multiplication rules (e.g., matrix rings, Clifford algebras).
\end{itemize}

\subsection*{Fields Defined by Irreducible Polynomials}

Finite fields of size \( p^n \), denoted \( \mathbb{F}_{p^n} \), can be constructed by taking the quotient ring:
\[
\mathbb{F}_{p^n} \cong \mathbb{F}_p[x]/(f(x)),
\]
where \( f(x) \) is an irreducible polynomial of degree \( n \) over \( \mathbb{F}_p \). Arithmetic in this field is performed modulo both \( p \) and \( f(x) \), yielding a fully multiplicative field of characteristic \( p \) with \( p^n \) elements~\cite{lidl1997finite}.

Because invariant-based constructions rely only on algebraic expressions (e.g., polynomial evaluation, cross-ratios, discriminants), all core operations extend naturally to such fields. The exponential terms \( p^t \) (or \( \alpha^t \)) can be reinterpreted in multiplicative subgroups of \( \mathbb{F}_{p^n}^\times \), and oscillatory components or masking maps can be adapted accordingly.

\subsection*{Impact on Security and Complexity}

Choosing different field structures affects both security properties and implementation trade-offs:

\begin{itemize}
  \item Larger fields (\( n > 1 \)) increase the entropy of sampled values and expand the space of possible invariants;
  \item Certain attacks (e.g., interpolation, root recovery, brute-force over \(\mathbb{Z}_M\)) become less effective when computations occur in high-order fields;
  \item When working over extension fields, the complexity of inversion and discriminant reconstruction may grow, providing additional hardness.
\end{itemize}

Moreover, embedding schemes into algebraic field extensions may allow encoding additional information in component dimensions (e.g., in field traces, conjugates, or norms), opening the door to richer cryptographic encodings.

\subsection*{Algebraic Extensions and Multidimensional Structures}

Beyond classical fields, invariant-based logic can be extended to finite-dimensional algebras~\cite{dummit2004abstract} over a base ring. For example:
\begin{itemize}
  \item In matrix rings \( \mathrm{Mat}_n(\mathbb{F}_p) \), one may define invariant relations based on determinants, traces, or symmetrized characteristic polynomials;
  \item In Clifford algebras or group rings, one may construct function values \( s(t) \in A \), where \( A \) is a noncommutative or graded algebra, and define invariants via structure-preserving operations;
  \item \textbf{Coordinate algebras and structured invariants:} \\
Invariant-based schemes are not limited to simple numerical fields or rings; they can be naturally extended to more abstract algebraic settings, such as coordinate algebras. These structures, typically written as \( \mathbb{F}_p[x_1, \dots, x_n]/I \), arise when one considers polynomial functions modulo a system of equations represented by an ideal \( I \). In this context, each element corresponds to a function defined on the set of solutions to that system—that is, on the points of an algebraic variety determined by \( I \).

For example, the ring \( \mathbb{F}_p[x, y]/(x^2 + y^2 - 1) \) describes functions over the set of solutions to the equation \( x^2 + y^2 = 1 \) in \( \mathbb{F}_p \). Function values in this ring are equivalence classes of polynomials, and arithmetic is performed modulo the relation \( x^2 + y^2 = 1 \). This framework is common in algebraic geometry and is used to model geometric structures such as curves, surfaces, and higher-dimensional varieties over finite fields.

When invariant-based constructions are defined over such coordinate algebras, the function evaluations \( s(t) \) may take values not in a field or modular ring, but in a structured algebraic environment. The associated invariants can then encode geometric or combinatorial constraints. For example:
\begin{itemize}
  \item A collection of points \( s_1, s_2, s_3, s_4 \in \mathbb{F}_p[x, y]/I \) might satisfy an invariant relation only if they lie on a common curve defined by \( I \);
  \item The invariant might test whether a symbolic "area" or determinant computed from the point coordinates vanishes (e.g., collinearity or planarity conditions);
  \item Alternatively, a cross-ratio could be defined symbolically, provided the denominators involved are invertible within the algebra.
\end{itemize}

However, care must be taken in such settings: unlike in a field, not every nonzero element in a coordinate algebra has a multiplicative inverse. As a result, invariant expressions involving division (like the cross-ratio) are only well-defined when certain elements are known to be invertible. In practice, this means:
\begin{itemize}
  \item Invariants can be defined only on subsets of the coordinate algebra where the required inverses exist;
  \item Alternatively, computations can be lifted to the field of fractions of the coordinate ring, though this may introduce additional complexity;
  \item In some cases, approximate or symbolic invariants can still serve as structural constraints, even without full division.
\end{itemize}

These generalizations allow invariant-based cryptography to incorporate deeper algebraic structure—connecting cryptographic constraints with geometric or combinatorial interpretations. Though more abstract, such extensions may support novel use cases, particularly in protocols involving symbolic commitments, algebraic proofs, or function evaluation over structured domains.
\end{itemize}

Although such generalizations lie beyond the main scope of this work, they indicate that the foundational idea of invariant-based cryptography---ensuring security through structural coherence---is not confined to any specific ring or field, but can naturally extend to a broad range of computable algebraic systems.

\section{Usage Scenarios and Cryptographic Properties}

Invariant-based schemes support a wide range of symmetric cryptographic operations by encoding hidden values within structured algebraic identities. In all cases, a secret value---typically denoted \(v\), \(h\), or \(z_4\) depending on the scheme---is not transmitted directly, but is instead embedded into a relation that only legitimate parties can evaluate, verify, or invert.

This approach enables compact, self-validating message structures with no need for asymmetric primitives. Below, we describe several representative scenarios that illustrate how invariant-based constructions can be deployed.

\subsection*{1. Secure Parameter Exchange}

Two parties share a common secret \(S\) (e.g., a session key or device binding).

\begin{itemize}
  \item Sender embeds a private value \(v\) into an invariant-preserving structure (e.g., roots of a polynomial, masked points, or aligned function evaluations).
  \item The message includes auxiliary data (e.g., \((a_2, a_3, y)\) or \((\tilde{z}_1, \tilde{z}_2, \tilde{z}_3)\)) sufficient for the receiver to reconstruct or verify \(v\).
  \item An integrity tag (such as \(H_{\text{check}}\)) binds the transmitted values to \(S\) and session context.
\end{itemize}

This supports authenticated, replay-resistant exchange of per-message parameters with minimal overhead.

\subsection*{2. Commitment to a Hidden Object}

Invariant-based schemes can serve as a lightweight commitment mechanism.

\begin{itemize}
  \item A party computes \(v := H(\text{object})\), embeds it into an invariant-based encoding using secret \(S\).
  \item The encoding (without \(v\)) is published or sent.
  \item At a later time, \(v\) is revealed; anyone with \(S\) can verify consistency via the invariant relation and hash.
\end{itemize}

This realizes a publicly verifiable commitment scheme without requiring trapdoor assumptions.

\subsection*{3. Split Trust and Cooperative Recovery}

The invariant structure can be split across parties to enforce collaboration.

\begin{itemize}
  \item Sender constructs an encoding where different components (e.g., \(s_1\) and \(s_3\), or \(a_2\) and \(a_3\)) are sent to different recipients.
  \item No single recipient can recover the hidden value alone.
  \item Only when components are combined and processed under the invariant does reconstruction succeed.
\end{itemize}

Such setups support threshold-style access control or co-signed message authorization.

\subsection*{4. Challenge-Response without Secret Disclosure}

Invariant-based constructions enable lightweight authentication via a challenge-response protocol, where the prover demonstrates knowledge of a shared secret \( S \) without revealing it or any derived values explicitly.

\begin{itemize}
  \item \textbf{Setup:} The verifier (e.g., server) selects a random session nonce \( z \) and some structural parameters (e.g., \( u \), polynomial roots \( a_2, a_3 \), or masked points \(\tilde{z}_1, \tilde{z}_2, \tilde{z}_3\)) that partially define an invariant instance.
  
  \item \textbf{Challenge:} These values are sent to the prover (e.g., client), who knows the shared secret \( S \).

  \item \textbf{Response:} The prover reconstructs the hidden part of the structure using \( S \) (e.g., computes the missing root \( a_1 \), offset \( h \), or cross-ratio point \( z_4 \)), and returns either:
  \begin{itemize}
    \item The missing value (e.g., \( h \)), or
    \item A hash binding (e.g., \( H(S, z, h) \)) that confirms knowledge without revealing \( h \) directly.
  \end{itemize}

  \item \textbf{Verification:} The verifier checks consistency with the invariant and validates the hash (if present).
\end{itemize}

This interaction authenticates the prover by confirming structural correctness of the invariant instance, which is only reconstructible with access to \( S \). No direct transmission of the secret or even the embedded value \( v \) is required.

\subsection*{5. Forward-Secure Ephemeral Derivation}

Invariant-based schemes can generate ephemeral secrets on a per-session basis, providing forward secrecy by construction. The key idea is to embed a temporary value (e.g., offset \( h \), embedded secret \( v \), or hidden point \( z_4 \)) into an invariant structure derived from a fresh session nonce \( z \).

\begin{itemize}
  \item Each session uses a new \( z \), ensuring that derived parameters (e.g., \( t(S, z) \)) and resulting invariants differ.
  \item The hidden value is recoverable or verifiable only by parties sharing the secret \( S \), but unlinkable to values in other sessions.
  \item Even if a single session is compromised, prior and future sessions remain protected.
\end{itemize}

This enables forward-secure symmetric coordination: lightweight, stateless, and resilient to partial compromise. Invariant reuse modes (as discussed earlier) allow minimizing bandwidth while preserving derivation uniqueness across time.

\subsection*{6. Stateless Stream Generation from Weak PRNGs}

In the \emph{Shared Root Mode}, the root \( a_1 \) is established at session initialization and reused throughout the session. Once this shared value is agreed upon, the sender (e.g., Alice) can generate a stream of invariant-preserving tuples \(\langle a_2, a_3, h \rangle\) by choosing new roots and offsets for each message.

Crucially, because each session is initialized with a unique nonce \( z \), and \( a_1 \) depends on this \( z \), even weak or deterministic generators (e.g., cyclic counters or linear feedback sequences) used to sample \( a_2, a_3, h \) will produce distinct sequences across sessions. That is:

\begin{itemize}
  \item Even if \( a_2, a_3, h \) are sampled from a fixed non-cryptographic generator, their interpretation and cryptographic embedding differ across sessions due to the influence of \( a_1 \) and \( t = t(S, z) \).
  \item The resulting value \( y = P(t + h) \), where \( P(x) = (x - a_1)(x - a_2)(x - a_3) \), is structurally randomized and unlinkable between sessions.
  \item If additional masking such as \( f(z) \in \mathrm{PGL}_2 \) is used to perturb the transmitted components, the stream becomes further obfuscated.
\end{itemize}

This behavior implies that within-session generation of ephemeral values can remain stateless and efficient without compromising isolation between sessions. Invariant-based design thus enables cryptographic separation through structural dependencies, rather than requiring cryptographically strong randomness at each step.

In practical deployments, this allows lightweight devices to use minimal or hardware-available entropy sources, relying on session-level rekeying and invariant enforcement to achieve security.

\subsection*{Remark}

While each specific scheme (discriminant-based, cross-ratio, or oscillatory) defines its own structure and recovery path, the overall cryptographic properties remain consistent:  
\emph{structural consistency} enforces authenticity;  
\emph{invariant preservation} ensures verifiability;  
\emph{masked construction} preserves confidentiality.

We emphasize that these scenarios represent only a subset of the possible applications. Invariant-based methods offer a versatile toolkit for protocol designers, especially in contexts where compactness, recoverability, and low computational cost are key design goals.

\section{Invariant-Based Puzzles and Constraint Embedding}

The invariant-based framework also opens the door to constructing structured cryptographic puzzles—challenge instances where solutions are constrained both by invariant relations and auxiliary algebraic conditions.

For example, consider the cross-ratio identity:

\[
\mathrm{CR}(z_1, z_2;\ z_3, z_4) = I,
\]

where \( z_1, z_2 \in \mathbb{Z}_M \) and invariant \( I \in \mathbb{Z}_M^\times \) are known. The values \( z_3, z_4 \) are unknown and must be chosen to satisfy the relation. Suppose further that we impose an external constraint such as:

\[
z_3^{z_4} \equiv k \mod M,
\]

for a known value \( k \in \mathbb{Z}_M \). Then the problem becomes: find values \( z_3, z_4 \in \mathbb{Z}_M \) satisfying both the invariant and the external equation.

Such puzzles may serve as:

\begin{itemize}
    \item Verifiable proofs of work or time;
    \item Structured commitments with embedded computational difficulty;
    \item Interactive challenges in protocols requiring controlled verification.
\end{itemize}

The ability to encode constraints into algebraically masked invariant structures provides a novel mechanism for generating tamper-resistant challenges and hybrid algebraic puzzles with adjustable difficulty.

\section{Outlook: Toward a Taxonomy of Invariants}

Across this work and the foundational study~\cite{semenov2025invariant}, we have introduced and analyzed three distinct cryptographic constructions based on algebraic invariants:
\begin{itemize}
    \item A functional four-point invariant over oscillatory evaluations;
    \item A discriminant-based scheme using the structure of cubic polynomials;
    \item A cross-ratio scheme derived from projective geometry and masked under $\mathrm{PGL}_2(\mathbb{Z}_M)$ transformations.
\end{itemize}

Each of these examples demonstrates how invariant-preserving relations can support verifiable, recoverable, and tamper-resistant symmetric protocols. The diversity of these constructions indicates that invariant-based cryptography is not limited to a narrow class of functional identities, but instead offers a flexible and expressive design space.

While this article does not aim to exhaustively classify all possible invariants, it sets a direction for future exploration. Below, we outline several additional categories of invariants that may be suitable for cryptographic adaptation:

\subsection*{1. Classical Algebraic Invariants}
\begin{itemize}
    \item \textbf{Multilinear and bilinear forms:} e.g., $x^T A y$, scalar products $\langle x, y \rangle$; invariant under orthogonal or unitary transformations. These may enable schemes involving masked vector representations or encoded inner products.
    \item \textbf{Symmetric functions:} elementary symmetric polynomials, power sums, and Newton identities; naturally invariant under permutations of inputs. Such structures may be applicable to pseudorandom multi-point encodings.
\end{itemize}

\subsection*{2. Geometric and Projective Invariants}
\begin{itemize}
    \item \textbf{Barycenters and means:} simple averaging invariants such as the center of mass; robust under input reordering.
    \item \textbf{Determinants and geometric volumes:} expressions such as $\det(x_1, x_2, x_3)$ represent oriented areas, volumes, or hypervolumes depending on dimensionality; they serve as multilinear invariants under linear transformations and can be used to construct multi-point schemes with geometric consistency.
\end{itemize}

\subsection*{3. Dynamical and Differential Invariants}
\begin{itemize}
    \item \textbf{Energy or action integrals:} preserved quantities in dynamical systems, such as Hamiltonians or Lagrangians; may be encoded via pseudorandom oscillator dynamics.
    \item \textbf{Lie invariants:} expressions invariant under group actions; applicable when function evaluations $s(t)$ exhibit symmetry under transformations of $t$.
\end{itemize}

\subsection*{4. Information-Theoretic Invariants}
\begin{itemize}
    \item \textbf{Entropy and mutual information:} statistical invariants across message ensembles; useful in generalized probabilistic or multi-session settings.
    \item \textbf{Divergence measures:} e.g., Kullback–Leibler divergence, used to detect tampering or inconsistency in distributions.
\end{itemize}

\subsection*{5. Group-Based and Cryptographic Invariants}
\begin{itemize}
    \item \textbf{Group-theoretic identities:} such as $g^a \cdot g^b = g^{a+b}$; may support schemes in which $s(t) \in \mathbb{G}$, a cryptographic group. This allows potential extensions to elliptic curves or pairings.
\end{itemize}

\subsection*{Perspective}

Rather than providing a definitive classification, we offer this outline as a guide for future development. Each invariant type brings its own structural assumptions, potential masking mechanisms, and verification logic. Exploring their cryptographic viability requires not only algebraic insight but also an understanding of how invariants behave under adversarial conditions, noise, and composition.

Invariant-based cryptography invites a new axis of design: not secrecy through concealment, but integrity through structural preservation. The examples developed here form the initial layer of a broader theory, whose future forms remain to be shaped.

\section*{Limitations and Open Questions}

While invariant-based cryptography offers a promising alternative to traditional symmetric designs, several limitations and risks should be acknowledged:

\begin{itemize}
    \item \textbf{Scheme-specific security analysis:} While the framework relies on structural indistinguishability and the difficulty of invariant recovery, each concrete instantiation requires dedicated cryptographic analysis to establish formal security guarantees. This article presents the conceptual foundation and illustrative constructions, leaving detailed hardness reductions and adversarial modeling for future work.
    
    \item \textbf{Algebraic fragility:} The schemes require precise algebraic consistency. Small implementation errors (e.g., in polynomial construction, masking, or inversion) may lead to silent failure or unintended leakage.
    
    \item \textbf{Side-channel susceptibility:} Like many symmetric primitives, these schemes may be vulnerable to side-channel analysis unless properly masked and implemented with constant-time operations.
    
    \item \textbf{Lack of standardization and review:} Invariant-based cryptographic primitives are novel and have not yet undergone extensive cryptanalytic scrutiny or standardization processes.
    
    \item \textbf{Dependency on field structure:} The algebraic behavior may vary significantly between prime fields, rings, and coordinate algebras, and certain invariants may not be well-defined or invertible in all settings.
\end{itemize}

Further study is required to formalize security guarantees, analyze compositional behavior, and evaluate performance under real-world constraints. We regard this work as an initial step in exploring the potential and limitations of invariant-based symmetric cryptography.

\section*{Conclusion}

This work advances the theory of invariant-based cryptography by establishing a general framework and demonstrating its viability across several distinct constructions. Starting from a functional identity over oscillatory sequences, and expanding to algebraic and geometric invariants such as polynomial discriminants and projective cross-ratios, we have shown that invariant relations can serve as effective cryptographic primitives.

These examples confirm that invariants can enforce structural coherence, resist forgery, and support lightweight symmetric protocols without reliance on classical one-way functions. While many questions remain open, the central message is clear: algebraic preservation can be as powerful as algebraic opacity. Invariant-based schemes offer a compact, flexible, and conceptually rich foundation for future cryptographic design.

% \bibliographystyle{plain}
% \bibliography{references}

\end{document}